\documentclass[11pt,table]{article}
\pdfoutput=1
\usepackage{fullpage} 
\usepackage{microtype}
\usepackage[T1]{fontenc}
\usepackage[utf8]{inputenc}
\usepackage{lmodern}

\usepackage{tabularx}
\usepackage{ucs}
\usepackage{amssymb}
\usepackage{amsthm}
\usepackage{amsmath,todonotes}
\usepackage{latexsym}
 \usepackage[english]{babel}
\usepackage{graphicx}
\usepackage{wrapfig}
\usepackage{caption}
\usepackage{subcaption}
\usepackage{txfonts}
\usepackage{mathrsfs}
\usepackage{algorithm}
\usepackage{dsfont}
\usepackage{enumerate}
\usepackage{multicol}
\usepackage{xspace}
\usepackage{csquotes}
\usepackage{authblk} 
\usepackage[absolute]{textpos}

\definecolor{darkblue}{rgb}{0,0,0.45}
\definecolor{darkred}{rgb}{0.6,0,0}
\definecolor{darkgreen}{rgb}{0.13,0.5,0}
\usepackage{hyperref}
\hypersetup{colorlinks,linkcolor=darkblue,citecolor=darkgreen,urlcolor=darkblue}
\usepackage{xcolor}

\newcommand{\Gdiamond}[0]{\ensuremath{K_4^-}\xspace}
\newcommand{\house}[0]{\ensuremath{C_5^+}\xspace}

\title{Flexibility of Planar Graphs---Sharpening the Tools to Get Lists of Size Four}

\author[1]{Ilkyoo Choi}
\author[2]{Felix Christian Clemen}
\author[3]{Michael Ferrara}
\author[4]{\\ Paul Horn}
\author[5]{Fuhong Ma}
\author[6,7,8]{Tomáš Masařík}

\affil[1]{Hankuk University of Foreign Studies}
\affil[ ]{\texttt{ilkyoo@hufs.ac.kr}}
\affil[2]{University of Illinois at Urbana-Champaign}
\affil[ ]{\texttt{fclemen2@illinois.edu}}
\affil[3]{University of Colorado Denver}
\affil[ ]{\texttt{michael.ferrara@ucdenver.edu}}
\affil[4]{University of Denver}
\affil[ ]{\texttt{paul.horn@du.edu}}
\affil[5]{Shandong University, Jinan, China}
\affil[ ]{\texttt{mafuhongsdnu@163.com}}
\affil[6]{Charles University, Prague, Czech Republic}
\affil[7]{University of Warsaw, Poland}
\affil[8]{Simon Fraser University, Burnaby, Canada}
\affil[ ]{\texttt{masarik@kam.mff.cuni.cz}}

\newtheorem{theo}{Theorem}

\newtheorem{lemma}[theo]{Lemma}

\newtheorem{question}[theo]{Question}
\newtheorem{obs}[theo]{Observation}

\newtheorem{defn}[theo]{Definition}

\theoremstyle{remark}

\newcommand{\brm}[1]{\operatorname{#1}}
\DeclareMathOperator\df{:=}

\newcommand{\ha}[0]{\cellcolor{blue!25}}
\newcommand{\eah}[0]{\cellcolor{green!25}}

\begin{document}
\date{}
	\maketitle
\begin{textblock}{20}(0, 12.5)
\includegraphics[width=40px]{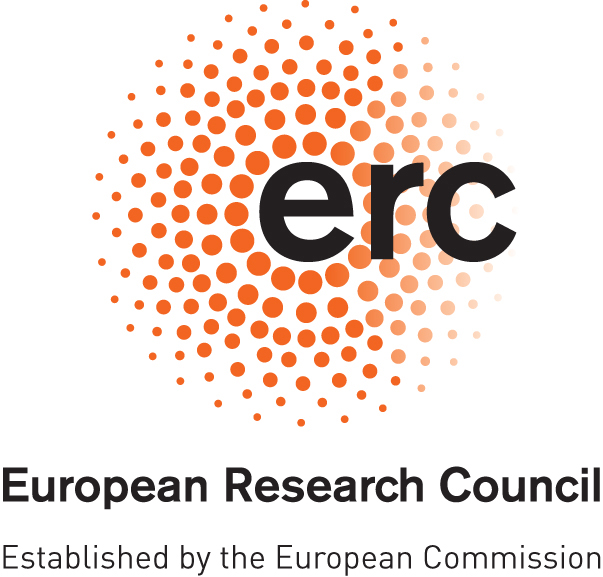}%
\end{textblock}
\begin{textblock}{20}(-0.25, 12.9)
\includegraphics[width=60px]{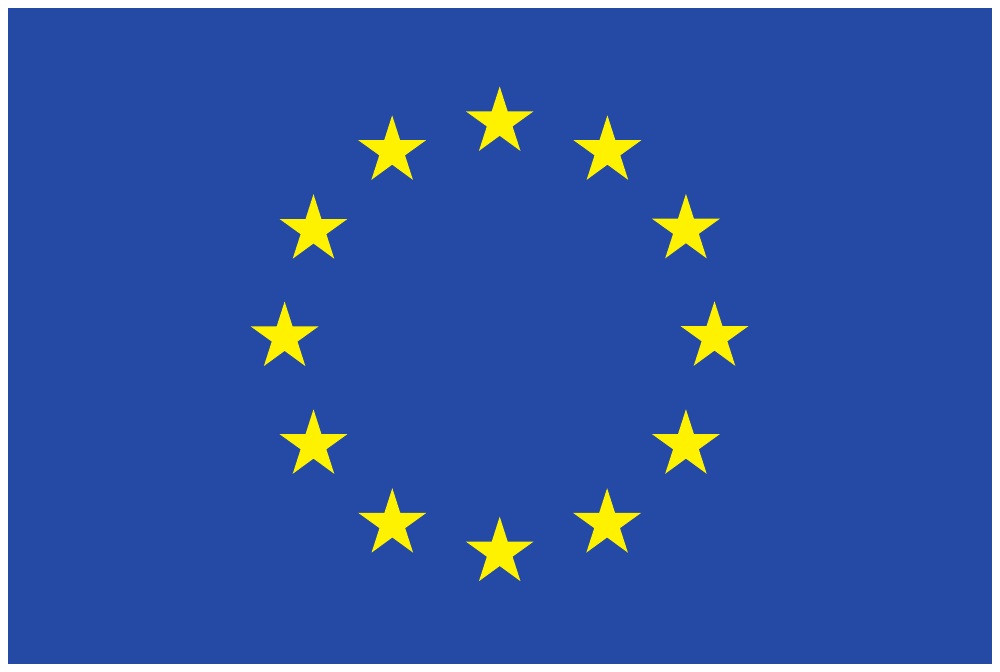}%
\end{textblock}

\begin{abstract}
  A graph where each vertex $v$ has a list $L(v)$ of available colors is {\it $L$-colorable} if there is a proper coloring such that the color of $v$ is in $L(v)$ for each $v$. A graph is {\it $k$-choosable} if every assignment $L$ of at least $k$ colors to each vertex guarantees an $L$-coloring. Given a list assignment $L$, an \textit{$L$-request} for a vertex $v$ is a color $c\in L(v)$. 
  In this paper, we look at a variant of the widely studied class of precoloring extension problems from [Z. Dvo\v{r}\'ak, S. Norin, and L. Postle: List coloring with requests. J. Graph Theory 2019], wherein one must satisfy "enough", as opposed to all, of the requested set of precolors.  A graph $G$ is {\it $\varepsilon$-flexible for list size $k$} if for any $k$-list assignment $L$, and any set $S$ of $L$-requests, there is an $L$-coloring of $G$ satisfying $\varepsilon$-fraction of the requests in $S$.  

It is conjectured that planar graphs are $\varepsilon$-flexible for list size $5$, yet it is proved only for list size $6$ and for certain subclasses of planar graphs.  
We give a stronger version of the main tool used in the proofs of the aforementioned results.   
By doing so, we improve upon a result by Masa\v{r}\'ik and show that planar graphs without $K_4^-$ are $\varepsilon$-flexible for list size $5$.
We also prove that planar graphs without $4$-cycles and $3$-cycle distance at least 2 are $\varepsilon$-flexible for list size $4$. 
Finally, we introduce a new (slightly weaker) form of $\varepsilon$-flexibility where each vertex has exactly one request.
In that setting, we provide a stronger tool and we demonstrate its usefulness to further extend the class of graphs that are $\varepsilon$-flexible for list size $5$. 
\end{abstract}

\section{Introduction}


A proper (vertex) coloring of a graph $G$ is an assignment of colors to the vertices of $G$ such that adjacent vertices receive distinct colors.  A widely studied class of problems in numerous branches of chromatic graph theory is the family of \textit{precoloring extension} problems, wherein the goal is to extend some partial coloring of the graph to a coloring with a desired property.  This general notion was introduced in \cite{PCI,PCII,PCIII}, and has been studied in a breadth of settings (see, for instance, \cite{albertson98listextend,AlbHutchPCplanar, albertson2004precoloring, AlbWestCircPC06, AxHutchListPC11, BrewNoelCircPC12, FGHSW, PruVoPC09,TuzaPCsurvey}).
Graph coloring with preferences have many application in various fields of computer science, such as scheduling~\cite{schedule}, register allocation~\cite{registry}, resource management~\cite{resource} and many others.


Motivated by the work of Dvořák and Sereni in \cite{dotri},
Dvořák, Norin, and Postle~\cite{Dvorakmain}
introduced the related concept of \emph{flexibility}.
The question of interest is the following: if some vertices of the graph have a preferred color, then is it possible to properly color the graph so that at least a constant fraction of the preferences are satisfied?  While this is not a precoloring extension problem in the classical sense, the idea of retaining a set of preferred, as opposed to prescribed, colors establishes a clear link between these problems.
This question is trivial in the usual proper coloring setting with a fixed number of colors, as we can permute the colors in a proper $k$-coloring of a graph in order to satisfy at least a $1/k$-fraction of the requests~\cite{Dvorakmain}.  On the other hand, in the setting of list coloring, the concept of flexibility gives rise to a number of interesting problems.

Before continuing, we present some formalities necessary for the results that follow.  A \emph{list assignment} $L$ for a graph $G$ is a function that assigns a set $L(v)$ of colors to each vertex $v\in V(G)$, and an \emph{$L$-coloring} is a proper coloring $\varphi$ such that $\varphi(v)\in L(v)$ for all $v\in V(G)$.
A graph $G$ is \emph{$k$-choosable} if $G$ is $L$-colorable from every list assignment $L$ where each vertex receives at least $k$ colors. 
The \emph{choosability} of a graph $G$ is the minimum $k$ such that $G$ is $k$-choosable.
A \emph{weighted request} is a function $w$ that assigns a non-negative real number to each pair $(v,c)$ where $v\in V(G)$ and $c\in L(v)$.  
Let $\displaystyle w(G,L)=\sum_{v\in V(G),c\in L(v)} w(v,c)$.
For $\varepsilon>0$, we say that $w$ is \emph{$\varepsilon$-satisfiable} if there exists an $L$-coloring $\varphi$ of $G$ such that
\[
 \sum_{v\in V(G)} w(v,\varphi(v))\ge\varepsilon\cdot w(G,L).
\]

An important special case is when at most one color can be requested at each vertex and all such colors have the same weight.
A \emph{request} for a graph $G$ with list assignment $L$ is a function $r$ with $\brm{dom}(r)\subseteq V(G)$ such that $r(v)\in L(v)$ for all $v\in\brm{dom}(r)$.
If each vertex requests exactly one color, i.e., $\brm{dom}(r)=V(G)$, then such a request is \emph{widespread}.
For $\varepsilon>0$, a request $r$ is \emph{$\varepsilon$-satisfiable} if there exists an $L$-coloring $\varphi$ of $G$ such that at least $\varepsilon|\brm{dom}(r)|$ vertices $v$ in $\brm{dom}(r)$ receive color $r(v)$.  In particular, a request $r$ is $1$-satisfiable if and only if the precoloring given by $r$ extends to an $L$-coloring of $G$.

We say that a graph $G$ with list assignment $L$ is \emph{$\varepsilon$-flexible}, \emph{weakly  $\varepsilon$-flexible}, and \emph{weighted $\varepsilon$-flexible} if every request, widespread request, and weighted request, respectively, is $\varepsilon$-satisfiable.
\footnote{Note that one can define a weighted variant of the weak flexibility by allowing only such requests that the sum of the weights for each vertex is exactly $1$, as was pointed out by an anonymous reviewer.
However, we will not consider this variant in the paper.}
If $G$ is (weighted/weakly) $\varepsilon$-flexible for every list assignment with lists of size $k$, then $G$ is {\it (weighted/weakly) $\varepsilon$-flexible for lists of size $k$.}

The main meta-question is whether there exists a universal constant $\varepsilon>0$ such that all $k$-choosable graphs are (weighted) $\varepsilon$-flexible for lists of size $k$.
In the paper that introduced flexibility, Dvořák, Norin, and Postle~\cite{Dvorakmain} established some basic properties and proved several theorems in terms of degeneracy and maximum average degree.
 A graph $G$ is $\emph{d-degenerate}$ if every subgraph has a vertex of degree at most $d$, and the \emph{degeneracy} of $G$ is the minimum $d$ such that $G$ is $d$-degenerate. 
\begin{theo}[Dvo{\v{r}}{\'a}k, Norin, and Postle \cite{Dvorakmain}]\label{thm:degeneracy}
For every $d\ge 0$, there exists $\varepsilon>0$ such that $d$-degenerate graphs are weighted $\varepsilon$-flexible for lists of size  $d+2$. 
\end{theo}
Compare this result with the corresponding trivial greedy bound for choosability; namely, $d$-degenerate graphs are $d+1$ choosable.
Since planar graphs are $5$-degenerate, Theorem~\ref{thm:degeneracy} implies that
there exists $\varepsilon>0$ such that every planar graph is weighted $\varepsilon$-flexible with lists of size $7$.
In the same paper~\cite{Dvorakmain}, two bounds in terms of the maximum average degree were developed, one of which implies that there exists $\varepsilon>0$ such that every planar graph is $\varepsilon$-flexible with lists of size $6$.
Since planar graphs are $5$-choosable~\cite{ch-general} and there exists a planar graph that is not $4$-choosable~\cite{voigt1993}, the following question is natural:

\begin{question}\label{q:general}
Does there exist $\varepsilon>0$ such that every planar graph is (weighted) $\varepsilon$-flexible for lists of size 5?
\end{question}

Dvořák, Masařík, Musílek, and Pangrác \cite{req-triangle} showed that planar graphs of girth at least 4 are weighted $\varepsilon$-flexible for lists of size $4$.
This is tight as planar graphs of girth at least 4 are $3$-degenerate, and hence $4$-choosable, and there exists a planar graph of girth 4 that is not 3-choosable~\cite{glebov,voigt1995}.
They also showed that planar graphs of girth at least 6 are weighted $\varepsilon$-flexible for lists of size $3$~\cite{req-six}.
There is still a gap, in terms of girth constraints, since planar graphs of girth at least 5 are 3-choosable~\cite{ch-girth5}.  

Masařík made further progress towards Question~\ref{q:general} by showing that planar graphs without\footnote{Note that in the whole paper \emph{without} refers to an ordinary subgraph relation.} 4-cycles are weighted $\varepsilon$-flexible for lists of size 5~\cite{req-four}.
This raises the natural corresponding question for planar graphs without 4-cycles, which are known to be 4-choosable~\cite{chsq-free}.

\begin{question}\label{q:c4}
Does there exist $\varepsilon>0$ such that every planar graph without $C_4$ is (weighted) $\varepsilon$-flexible for lists of size 4?
\end{question}

\subsection{Our Results}

We prove a strengthening of the result in~\cite{req-four} towards solving Question~\ref{q:general}.
Wang and Lih~\cite{2002WaLih} conjectured that planar graphs without \Gdiamond ($K_4$ without an edge, also known as a diamond) are $4$-choosable. This conjecture remains open. 
Improving the result in~\cite{req-four}, we prove that planar graphs without \Gdiamond are weighted $\varepsilon$-flexible for lists of size $5$. 
This is the largest subclass of planar graphs that is known to be weighted $\varepsilon$-flexible for lists of size $5$. 

\begin{theo}
\label{diamond}
There exists $\varepsilon > 0$ such that every planar graph without $K_4^-$ is weighted $\varepsilon$-flexible for lists of size 5.
\end{theo}

We also investigate Question~\ref{q:c4}; $\varepsilon$-flexibility for two subclasses of planar graphs without $4$-cycles, and show that lists of size 4 are actually sufficient.

\begin{theo}
\label{diamondtriangle}
There exists $\varepsilon > 0$ such that every planar graph without $C_4$ and with $C_3$ distance at least 2 is weighted $\varepsilon$-flexible for lists of size 4.
\end{theo}

\begin{theo}\label{C4C5}
There exists $\varepsilon > 0$ such that every planar graph without $C_4, C_5, C_6$ is weighted $\varepsilon$-flexible for lists of size 4.
\end{theo}

In order to prove the Theorems~\ref{diamond}, \ref{diamondtriangle}, and \ref{C4C5}, we strengthened the main tool from~\cite{Dvorakmain}, which was explicitly presented in~\cite{req-triangle}.
This tool, Lemma~\ref{lem:evenImproved} in Section~\ref{s:mainTool}, allows us to construct more fine-tuned reducible configurations, thereby reducing the complexity of the discharging argument.

The concept of weak flexibility enables stronger reducible configurations when widespread requests are considered.  To identify such configurations, we adapted Lemma~\ref{lem:evenImproved} for widespread requests.
Using the notion of weak flexibility, we derive the following result.
For an example of a configuration that is possible in this setting, but not in that of general flexibility, see (RC3) in Section~\ref{sec:weak}.

As planar graphs without $C_4$ are not $3$-degenerate, the following result further extends our attempts to attack Question~\ref{q:general}.
In this case \Gdiamond is allowed and we instead forbid the ``house" \house (a $C_3$ and a $C_4$ sharing an edge) and $K_{2,3}$.

\begin{theo}\label{weak}
  There exists $\varepsilon > 0$ such that every planar graph with neither \house nor $K_{2,3}$ is weakly $\varepsilon$-flexible for lists of size 5.
\end{theo}

Table~\ref{tab:flex-choose} compares known related results in terms of flexibility and choosability regarding planar graphs with forbidden structures. 
  The results proved in this paper are highlighted in green. 
 
  \begin{table}[!htbp]
\begin{center}
  \begin{tabular}{ | l|| c ||  c | c | c |}
    \hline
    Planar graphs &  & $C_3$ & $C_3$, $C_4$ & $C_3$, $C_4$, $C_5$ \\
    \hline\hline
    Degeneracy & 5\footnotemark[1] & 3\footnotemark[1]        & 3\footnotemark[1]  & 2   \\
    Choosability 
                              & 5 \cite{ch-general,voigt1993}      & 4~\cite{glebov,voigt1995}      &  3~\cite{ch-girth5}  & 3  \\
    Weak Flexibility
                              &\ha 6            & 4         &\ha 4  & 3 \\
    Flexibility   &\ha 6 \cite{Dvorakmain}           & 4         &\ha 4  & 3  \\
    Weighted Flexibility
                 &\ha 7
                 & 4~\cite{req-triangle}         &\ha 4  & 3~\cite{req-six}  \\
\hline
  \end{tabular}\\\smallskip
  \begin{tabular}{ | l|| c |  c | c | c || c | c | }
    \hline
    Planar graphs &  $C_4$  & $C_4$,$C_3$ dist. $\leq 1$ & $C_4$, $C_5$ & $C_4$, $C_5$, $C_6$ & $K_4^-$& \house, $K_{2,3}$\\
    \hline\hline
    Degeneracy &  4\footnotemark[1] & 3\footnotemark[1]$^{,}$\footnotemark[2]  & 3\footnotemark[1]~\cite{WangLih} & 3\footnotemark[1] & 4\footnotemark[1] &4\footnotemark[1] \\
    Choosability &   4~\cite{chsq-free} &\eah 4 & 4~\cite{V07}&\eah 4  & \eah 5 &\eah 5 \\
    Weak Flexibility
                 & \ha 5    &\eah 4  &\ha 5  &\eah 4 & \eah 5 & \eah 5 T\ref{weak} \\
    Flexibility   &\ha  5 &\eah 4  &\ha 5 &\eah 4 &\eah 5&\ha 6 \\
    Weighted Flexibility
                       &\ha  5~\cite{req-four}&\eah 4~T\ref{diamondtriangle}  &\ha 5 &\eah 4~T\ref{C4C5} &\eah 5~T\ref{diamond} &\ha 6\\
\hline
  \end{tabular}
\end{center}
\caption{\label{tab:flex-choose}
  Known results in terms of degeneracy, flexibility, and choosability, including known lower bounds for choosability.
  Each graph class is a subclass of planar graphs. The first row indicates the forbidden subgraphs for each column. 
  Results without a reference follow from some other entry in the table, including the degeneracy arguments implied by Theorem~\ref{thm:degeneracy}.
  Entries in blue are prior results that are not known to be tight, and entries in green are from this paper but are not known to be tight.  
}
\end{table}
\addtocounter{footnote}{1}
\footnotetext{Note that the dodecahedron is a planar graph without $C_3,C_4$ that is not 2-degenerate.
The  icosidodecahedron is a planar graph without $C_4$ that is not 3-degenerate.
 The truncated cube is a planar graph without $C_4,C_5,C_6,C_7$ that is not 2-degenerate.
The icosahedron is a planar graph that is not 4-degenerate.
 }
  \stepcounter{footnote}
  \footnotetext{We have not found a reference for this result and therefore include a proof for completeness as Observation~\ref{obs:degeneracy} in Section~\ref{sec:diamondtriangle}}

In Section~\ref{s:mainTool}, we develop the notation  and prove Lemma~\ref{lem:evenImproved}, our main tool. We lay out our proof strategy in Section~\ref{sec:lemmas} and the proofs of Theorems~\ref{diamond}, \ref{diamondtriangle}, \ref{C4C5}, and \ref{weak} are in Sections~\ref{sec:diamond}, \ref{sec:diamondtriangle}, \ref{sec:C4C5}, and~\ref{sec:weak}, respectively.

\newcommand{\E}{{\mathbb{E}}}
\renewcommand{\L}{\mathcal{L}}
\newcommand{\e}{{\mathbb{E}}}
\newcommand{\asum}[1]{\widetilde{\sum_{#1}}}
\newcommand{\diag}{{\rm diag}}
\newcommand{\pp}{\mathbf{p}}
\newcommand{\Var}{{\rm Var}}
\newcommand{\Cov}{{\rm Cov}}
\newcommand{\diam}{{\rm diam}}
\newcommand{\p}{\mathbb{P}}

\section{Main tool}\label{s:mainTool}

In this section, we strengthen and modify the main tool used in~\cite{req-triangle,Dvorakmain}, in order to better understand the structural properties of a hypothetical minimum counterexample.
Let $1_I$ denote the characteristic function of $I$, meaning that $1_I(v)=1$ if $v\in I$ and $1_I(v)=0$ otherwise.
For integer functions on the set of vertices of $H$, we define addition and subtraction in the natural way,
adding and subtracting, respectively, their values at each vertex. Given a function $f:V(H)\to\mathbb{Z}$ and a vertex $v\in V(H)$, let $f\downarrow v$ denote the function such that $(f\downarrow v)(w)=f(w)$ for $w\neq v$
and $(f\downarrow v)(v)=1$.  A list assignment $L$ is an \emph{$f$-assignment} if $|L(v)|\ge f(v)$ for all $v\in V(H)$.
  Given a set of graphs $\mathcal F$ and a graph $H$, a set $I\subset V(H)$ is \emph{$\mathcal{F}$-forbidding} if the graph $H$ together with one additional vertex adjacent to all of the vertices in $I$ does not contain any graph from $\mathcal{F}$. 
  Let  $H$ be an induced subgraph of $G$ and $v\in V(H)$. We define $\deg_H(v)$ as the degree of $v$ measured in the induced subgraph $H$.
  
\begin{defn}[$(\mathcal{F},k)$-boundary-reducibility]
  A graph $H$ is an \emph{$(\mathcal{F},k)$-boundary-reducible} induced subgraph of $G$ if there exists a set $B\subsetneq V(H)$ such that
\begin{itemize}
\item[(FIX)] for every $v\in V(H)\setminus B$, $H- B$ is $L$-colorable for every (($k-\deg_G+\deg_{(H-B)})\downarrow$ $v$)-assignment $L$, and 
  \item[(FORB)] for every $\mathcal{F}$-forbidding set $I \subseteq V(H)\setminus B$ of size at most $k-2$, $H- B$ is $L$-colorable for every ($k-\deg_G+\deg_{(H-B)}-1_I$)-assignment $L$. 
\end{itemize}
\end{defn}\label{reduce}
We define an additional reducibility condition wherein we weaken the (FIX) property in order to consider a new, weaker variation of flexibility.  

\begin{defn}[weak $(\mathcal{F},k)$-boundary-reducibility]
  A graph $H$ is a \emph{weakly $(\mathcal{F},k)$-boundary-reducible} induced subgraph of $G$ if there exists a set $B\subsetneq V(H)$ and a non-empty set $\text{Fix}(H)\subseteq \left(V(H)\setminus B\right)$ such that
\begin{itemize}
  \item[(FIX)] for every $v\in \text{Fix}(H)$, $H- B$ is $L$-colorable for every (($k-\deg_G+\deg_{(H-B)})\downarrow$ $v$)-assignment $L$, 
    and 
  \item[(FORB)] for every  $\mathcal{F}$-forbidding set $I \subseteq V(H)\setminus B$ of size at most $k-2$, $H- B$ is $L$-colorable for every ($k-\deg_G+\deg_{(H-B)}-1_I$)-assignment $L$. 
\end{itemize}
\end{defn}

In general, we may allow $\text{Fix}(H)$ to be only one vertex, which clearly establishes this notion as a weaker one than that given in Definition \ref{reduce}. In both of the preceding definitions, we will sometimes refer to the set $B$ as the \textit{boundary} of the configuration.

\begin{defn}[$(\mathcal{F},k,b)$-resolution]\label{def:resolution}
  Let $G$ be a graph with lists of size $k$ that does not contain any graph in $\mathcal{F}$ as an induced subgraph.
We define \emph{$(\mathcal{F},k,b)$-resolution} of $G$ as a set $G_i$ of nested subgraphs for $0\leq i\leq M$, such that $G_0\df G$ and
\[
  G_i\df G- \bigcup_{j=1}^{i} {\left(H_j- B_j\right)},
\]
where each $H_i$ is an induced  $(\mathcal{F},k)$-boundary-reducible subgraph of $G_{i-1}$ with boundary $B_i$ such that $|V(H_i)\setminus B_i|\le b$ and $G_{M}$ is an $(\mathcal{F},k)$-boundary-reducible graph with empty boundary and size at most $b$. 
For technical reasons, let $G_{M +1}\df\emptyset$.
\end{defn}

A \emph{weak $(\mathcal{F},k,b)$-resolution} is defined anaologously to an $(\mathcal{F},k,b)$-resolution; it uses weak $(\mathcal{F},k)$-boundary-reducibility in the place of $(\mathcal{F},k)$-boundary-reducibility.  

It is our goal to show that every graph that does not contain any graph from ${\mathcal F}$ as a subgraph contains a (weakly) reducible subgraph.  Conceptually, we then think of a (weak) resolution as an inductively-defined object obtained by iteratively identifying some (weakly) reducible subgraph $H$ with boundary $B$ and deleting $H-B$ until $V(G)$ is exhausted.   

Going forward, when considering weak reducibility, let $\text{Fix}(G)$ denote the union of each $\text{Fix}(H)$ over the weak $(\mathcal{F},k)$-boundary-reducible subgraphs $H$ in some resolution of $G$.  While $\text{Fix}(G)$ depends on the particular resolution under consideration, we will generally omit mention of the resolution when the context is clear. By definition, the following property holds:
\[\text{Fix}(G)\cap \left(G_i- G_{i+1}\right)\neq\emptyset.\]
When considering $(\mathcal{F},k)$-boundary-reducibility, we take $\text{Fix}(H)=V(H)\setminus B$ where $B$ is the boundary of $H$.  

To prove weighted $\varepsilon$-flexibility, we use the following observation made by Dvo\v{r}\'ak et al.~\cite{Dvorakmain}.

\begin{lemma}[\cite{Dvorakmain}]\label{lemma:distrib}
Let $G$ be a graph and let $L$ be a list assignment on $V(G)$.
Suppose $G$ is $L$-colorable and there exists a probability distribution on $L$-colorings $\varphi$ of $G$ such that for every $v\in V(G)$ and $c\in L(v)$, $$\mathbf{Prob}[\varphi(v)=c]\ge\varepsilon.$$
Then $G$ with $L$ is weighted $\varepsilon$-flexible.
\end{lemma}

In light of lemma~\ref{lemma:distrib}, we can derive a similar lemma for weak $\varepsilon$-flexibility, provided that the assumptions hold for sufficiently large and evenly distributed set of vertices.

\begin{lemma}[]\label{lemma:distribWeak}
  Let $b$ be an integer.
Let $G$ be a graph with a widespread request and let $L$ be a list assignment on $V(G)$.
Suppose $G$ is $L$-colorable with a weak $(\mathcal{F},k,b)$-resolution and there exists a probability distribution on $L$-colorings $\varphi$ of $G$ such that for every $v\in \text{Fix}(G)$ and $c\in L(v)$, $\mathbf{Prob}[\varphi(v)=c]\ge\varepsilon$.
Then $G$ with $L$ is weakly $\left(\varepsilon\cdot \frac{1}{b}\right)$-flexible.
\end{lemma}
The proof is very similar to proof of Lemma~\ref{lemma:distrib}, but it makes use of the fact that requests are made for all vertices of $G$.
\begin{proof}
  Let $r$ be a widespread request for $G$ and $L$.
  Let $\phi$ be  chosen at random based on the given probability distribution.
  By assumption $|\text{{Fix}(G)}|\ge \frac{|V|}b$.
By linearity of expectation:
\[\mathbf{E}\Biggl[\sum_{v\in \text{Fix}(G),\phi(v)=r(v)} 1 \Biggr]=\sum_{v\in \text{Fix}(G)} \mathbf{Prob}[\phi(v)=r(v)]\ge\varepsilon \cdot\frac{|V|}{b},
\]
and thus there exists an $L$-coloring $\phi$ with $\left(\varepsilon\cdot \frac{|V|}{b}\right)$ satisfied requests.
\end{proof}

Now, we are ready to strengthen the key lemma implicitly presented by Dvo\v{r}\'ak, Norin, and Postle in~\cite{Dvorakmain}, and explicitly formulated as Lemma~4~in~\cite{req-triangle}).

\begin{lemma}\label{lem:evenImproved}
  For all integers $k\ge 3$ and $b\ge 1$ and for all sets $\mathcal{F}$ of forbidden subgraphs there exists an $\varepsilon>0$ as follows. 
Let $G$ be a graph with an $(\mathcal{F},k,b)$-resolution. 
 Then $G$ with any assignment of lists of size $k$ is weighted $\varepsilon$-flexible.
 Furthermore, if the request is widespread and $G$ has weak $(\mathcal{F},k,b)$-resolution, 
 then $G$ with any assignment of lists of size $k$ is weakly $\left(\varepsilon\cdot \frac{1}{b}\right)$-flexible.
\end{lemma}

Even though the statement of the lemma is noticeably stronger, its proof remains quite similar to the original formulation. 
We include the proof for the sake of completeness.

 \begin{proof}
Let $p=k^{-b}$ and $\varepsilon=p^{k-1}$.
For a graph $G$ satisfying the assumptions and an assignment $L$ of lists of size $k$,
we prove the following claim by induction on the $(\mathcal{F},k,b)$-resolution:

\smallskip

There exists a probability distribution on $L$-colorings $\varphi$ of $G$ such that
\begin{itemize}
  \item[(i)] for every $v\in \text{Fix}(G)$ and a color $c\in L(v)$, the probability that $\varphi(v)=c$ is at least $\varepsilon$, and
\item[(ii)] for every color $c$ and every $\mathcal{F}$-forbidding set $I$ in $G$ of size at most $k-2$, the probability that $\varphi(v)\neq c$ for all $v\in I$ is at least $p^{|I|}$.
\end{itemize}
Part (i) with $\text{Fix}(G)=V(G)$ implies that $G$ with $L$ is weighted $\varepsilon$-flexible by Lemma~\ref{lemma:distrib}.
Part (i) with assumed widespread request implies that $G$ with $L$ is weakly $\left(\varepsilon\cdot \frac{1}{b}\right)$-flexible by Lemma~\ref{lemma:distribWeak}.\\\smallskip

The claim clearly holds for a graph with no vertices, the base case of the induction.
Hence, suppose $V(G_i)\neq\emptyset$.
By the assumptions, there exists a subgraph $H$ of $G$ such that $H$ is $(\mathcal{F},k)$-boundary-reducible.
By definition, there exists a boundary set $B\subset V(H)$, and let  $Q\df H - B$.
Moreover, by assumption, we know that the order of $Q$ is at most $b$.
By the induction hypothesis, there exists a probability distribution
on $L$-colorings of $G_{i+1}$ satisfying (i) and (ii).
Choose an $L$-coloring $\psi$ from this distribution and let $L'$ be the list assignment on $V(G[Q])$ defined by $L'(y)=L(y)\setminus \{\psi(v):v\in V(G)\setminus V(Q), vy\in E(G)\}$.
Note that $|L'(y)|\ge k-\deg_G(y)+ \deg_{G[Q]}(y)$
for all $y\in V(H)$, and thus $G[V(Q)]$ has an $L'$-coloring by (FORB) applied with $I=\emptyset$. Among all $L'$-colorings of $G[V(Q)]$, choose one
uniformly at random, extending $\psi$ to an $L$-coloring $\varphi$ of $G$.

Let us first argue  (ii) holds.  
Let $I_1=I\setminus V(Q)$  and $I_2=I\cap V(Q)$.  
By the induction hypothesis,
we have $\varphi(v)\neq c$ for all $v\in I_1$ with probability at least $p^{|I_1|}$.  
If $I_2=\emptyset$, then this implies (ii).
Hence, suppose $|I_2|\ge 1$.  
For $y\in I_2$, let $L_c(y)=L'(y)\setminus\{c\}$;
and for $y\in V(Q)\setminus I_2$, let $L_c(y)=L'(y)$.  
Note that $|L_c(y)|\ge k-\deg_G(y)+ \deg_{G[V(Q)]}(y)-1_I(y)$ for all $y\in V(Q)$. 
Thus by (FORB), $G[V(Q)]$ has an $L_c$-coloring.  
Since $G[V(Q)]$ has at most $k^b$ $L'$-colorings, we conclude that
the probability that $\varphi(y)\neq c$ for all $y\in I_2$ is at least $1/k^b=p\ge p^{|I_2|}$.
Hence, the probability that $\varphi(y)\neq c$ for all $y\in I$ is at least $p^{|I_1|+|I_2|}\ge p^{|I|}$, implying (ii).

Next, let us argue  (i) holds.  
For $v\in V(G)\setminus V(Q)$, this is true by the induction hypothesis.
Hence, suppose that $v\in V(Q)$, and let $I$ be the set of neighbors of $v$ in $V(G)\setminus V(Q)$.
Since $G$ does not contain any graph from $\mathcal{F}$ but does contain the boundary $B$, and all vertices in $I$ have a common neighbor, the set $I$ is $\mathcal{F}$-forbidden in $G-V(Q)$.
Furthermore, (FORB) implies $1\le k-\deg_G(v)+ \deg_{G[V(Q)]}(v)-1_{\{v\}}(v)=k-|I|-1$, and thus $|I|\le k-2$.
Hence, by the induction hypothesis, we have $\psi(u)\neq c$ for all $u\in I$ with probability at least $p^{k-2}$.  Assuming this is the case, (FIX) implies there exists an $L'$-coloring of $G[V(Q)]$ that  gives $v$ the color $c$.
Since $G[V(Q)]$ has at most $k^b$ $L'$-colorings, we conclude that the probability that $\varphi(v)=c$ is at least $p^{k-2}/k^b=\varepsilon$.
Hence, (i) holds.
 \end{proof}

\section{Common proof preliminaries}\label{sec:lemmas}

We gather common definitions and reducible configurations for the forthcoming proofs. 
We also give an overview of the discharging method. 

A {\it $d$-vertex}, a {\it $d^+$-vertex}, and a {\it $d^-$-vertex} are a vertex of degree $d$, at least $d$, and at most $d$, respectively. 
A {\it $d$-face}, a {\it $d^+$-face}, and a {\it $d^-$-face} are defined analogously. 
A {\it $(d_1, d_2, d_3)$-face} is a $3$-face where the degrees of the vertices on the face is $d_1, d_2, d_3$.  
Throughout all the figures in the paper, black vertices have all their incident edges drawn, whereas a white vertex may have more edges incident than drawn since white vertices are in the boundary. 

Going forward, we will use the number of colors ``available" to a vertex $v$ in a configuration $H$ as the maximum number of colors remaining in $L(v)$ after coloring vertices exterior to the configuration (more precisely outside $H- B$).  
When considering (FIX) we reduce the number of available colors on a ``fixed" vertex to 1, and when considering (FORB), we reduce the number of available colors on the vertices of a candidate set $I$ by 1.

  \begin{lemma}\label{lem:singleVertex}
  A $(k-2)^-$-vertex is $(\emptyset,k)$-boundary-reducible with empty boundary.
 \end{lemma}
  \begin{proof} 
 Let $H$ be the graph induced by a $(k-2)^-$-vertex $v$ and set the boundary to be empty. 

 (FIX) holds since $H$ has only one vertex $v$. 

 (FORB) also holds since there is still one available color for $v$. 
 \end{proof}
 
   \begin{lemma}\label{lem:twoVerticesOnTriangle}
     For $k\ge 4$, a $3$-cycle $uvw$ with two $(k-1)$-vertices $u,v$ is $(\{K_4^-\}),k)$-boundary-reducible with boundary $\{w\}$.
 \end{lemma}

 \begin{proof}
   Let $H$ be a $3$-cycle $uvw$ such that $\deg(u)=\deg(v)=k-1$ and set the boundary $B=\{w\}$.

 (FIX): Each of $u, v$ has two available colors. 
 Fix a color $\varphi(u)$ and choose an available color for $\varphi(v)$ that is not $\varphi(u)$ to extend the coloring. 
 Fixing a color for $v$ is symmetric. 
 
 (FORB): 
 Let $I\subset V(H)\setminus B$ where $|I|\leq k-2$. 
 If $|I|\geq 2$, then $I$ is not $\{K^-_4\}$-forbidding since connecting a new vertex to both $u$ and $v$ always creates a $K^-_4$. 
 It remains to consider the cases where $|I|=1$ but that is implied by (FIX).
\end{proof}

We use the discharging method for the proofs of our theorems. 
We end this section with a brief overview of the method, for more detailed introduction to discharging method for graph coloring see~\cite{Cranston}.
Given a theorem we aim to prove, let $H$ be a counterexample with the minimum number of vertices.
Fix a plane embedding of $H$ and let $F(H)$ denote the set of faces of $H$. 
To each $z\in V(H)\cup F(H)$, assign an {\it initial charge} $ch(z)$ so that the total sum is negative.
In our proofs, this part is a standard and straightforward aplication of Euler's formula. 
We then redistribute the charge according to some {\it discharging rules}, which will preserve the total charge sum. 
Let $ch^*(z)$ denote the {\it final charge} at each $z\in V(H)\cup F(H)$. 
We recount the charge at this point and show that the final charge is non-negative for each vertex and face to conclude that the sum of the final charge is non-negative. 
This is a contradiction since the initial charge sum is negative and the discharging rules preserve the total charge sum.
We conclude that a counterexample could not have existed. 

\section{Proof of Theorem \ref{diamond}}\label{sec:diamond}

In this section we prove Theorem~\ref{diamond}. Let  $\mathcal F=\{\Gdiamond\}$ and let $G$ be a counterexample to Theorem~\ref{diamond} with the minimum number of vertices. 
Fix a plane embedding of $G$ and note that by minimality, $G$ must be connected.
Let $L$ be a list assignment on $V(G)$ where each vertex receives at least five colors. 
The following configurations cannot appear in $G$ by Lemmas \ref{lem:singleVertex} and \ref{lem:twoVerticesOnTriangle}:
 \begin{itemize}
     \item [(RC1)] A $3^-$-vertex. 
     \item [(RC2)] A $3$-cycle incident with two $4$-vertices. 
 \end{itemize}
 
 \begin{figure}[H]
  \centering
  \includegraphics[height=3cm]{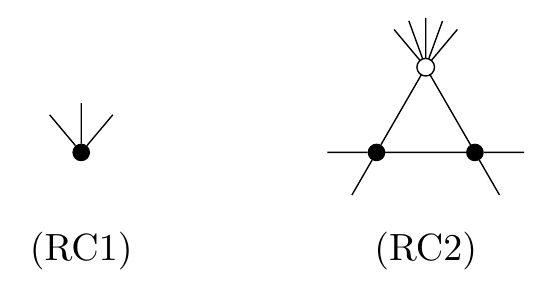}
  \caption{Reducible configurations}
  \label{fig:1}
\end{figure}

For each vertex $v$ and each face $f$, let $ch(v)=\deg(v)-4$ and $ch(f)=|f|-4$.
By Euler's formula the sum of initial charge is negative: $\sum_{v\in V(G)}(\deg(v)-4)+\sum_{f\in F(G)}(|f|-4)=2|E(G)|-4|V(G)|+2|E(G)|-4|F(G)|=-8$.
The following is the only discharging rule: 

\begin{itemize}
    \item [(D1)] Every $5^+$-vertex sends $1/2$ to each incident $3$-face.
\end{itemize}

We now check that each vertex and each face has non-negative final charge. Considering the vertices first, by (RC1) there is no $3^-$-vertex.  If $v$ is a $4$-vertex, then $ch^*(v)=ch(v)=\deg(v)-4=0$, since $v$ is not involved in the discharging rules. If $v$ is a $5^+$-vertex, then $ch^*(v)\geq \deg(v)-4-\lfloor\frac{\deg(v)}{2}\rfloor\cdot\frac{1}{2}\geq0$ by (D1).
Note that a vertex $v$ is incident with at most $\lfloor\frac{\deg(v)}{2}\rfloor$ $3$-faces since there is no \Gdiamond. 
Hence, each vertex has non-negative final charge. 

We next turn our attention to faces, again with the goal of showing that their final charge is non-negative.  If $f$ is a $4^+$-face, then $ch^*(v)=ch(v)=|f|-4\geq0$, since $f$ is not involved in the discharging rules. If $f$ is a $3$-face, then by (RC1) and (RC2), $f$ is incident with at least two $5^+$-vertices.  By (D1), $ch^*(f)\geq-1+2\cdot\frac{1}{2}=0$, completing the proof of Theorem \ref{diamond}.

\section{Proof of Theorem \ref{diamondtriangle}}\label{sec:diamondtriangle}

In this section, we prove Theorem~\ref{diamondtriangle}. 
Let $\mathcal F$ consist of $C_4$ and all possible configurations such that $C_3$'s are in distance at most $1$ and let $G$ be a counterexample to Theorem~\ref{diamondtriangle} with the minimum number of vertices.
Fix a plane embedding of $G$ and note that by minimality, $G$ must be connected. 
Let $L$ be a list assignment on $V(G)$ where each vertex receives at least four colors. 
The following configurations cannot appear in $G$:
\begin{itemize}
    \item [(RC1)] A $2^-$-vertex.
    \item [(RC2)] A $3$-vertex adjacent to two $3$-vertices. 
    \item [(RC3)] A $d$-vertex $v$ adjacent to $d-2$ vertices of degree 3 where $v$ is either on a $3$-cycle or is adjacent to a vertex on a $3$-cycle. 
\end{itemize}

\begin{figure}[H]
  \centering
    \includegraphics[height=3cm]{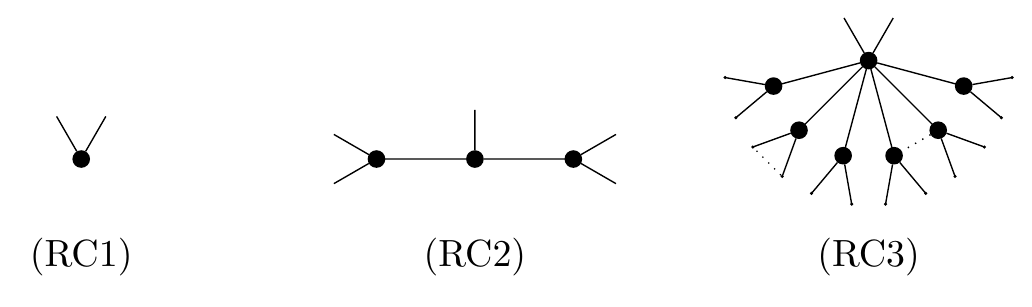}
  \caption{Reducible configurations: dotted lines indicate that there is a $3$-cycle somewhere}
\end{figure}

First, observe that (RC1) is reducible by Lemma~\ref{lem:singleVertex}.

\begin{lemma}
  \normalfont (RC2) \textit{ is $(\emptyset,4)$-boundary-reducible with empty boundary}.
 \end{lemma}
 
 \begin{proof}
Let $a, b, c$ be consecutive $3$-vertices on a path. 
Let $H$ be the graph on $a, b, c$ and set the boundary $B=\emptyset$. 
If $ac$ is an edge, then this configuration is reducible by Lemma~\ref{lem:twoVerticesOnTriangle}, so we may assume $a$ and $c$ are non-adjacent. 
Hence $b$ has three available colors and $a$ and $c$ both have two available colors.  

(FIX): When the color of $b$ is fixed as $\varphi(b)$, then the coloring can be extended by choosing colors for $a$ and $c$ that are different from $\varphi(b)$, which is possible since $a$ and $c$ have two available colors. 
When the color of $a$ is fixed as $\varphi(a)$, then the coloring can be extended by choosing a color $\varphi(b)$ for $b$ that is not $\varphi(a)$, and then choosing a color for $c$ that is not $\varphi(b)$. 
Fixing the color of $c$ is symmetric.

 (FORB): 
 Let $I\subset V(H)\setminus B$ where $|I|\leq 2$. 
  When $I=\{a, b\}$, $a, b, c$ have one, two, two, respectively, available colors. 
It remains to consider the cases where $|I|=1$ but that is implied by (FIX). 
An $L$-coloring of $H$ can be obtained by choosing an available color $\varphi(a)$ for $a$, choosing an available color $\varphi(b)$ for $b$ that is not $\varphi(a)$, and then choosing  an available color for $c$ that is not $\varphi(b)$.
  When $I=\{a, c\}$, $a, b, c$ have one, three, one, respectively, which is colorable.
It remains to consider the cases where $|I|=1$ but that is implied by (FIX). 
  \end{proof}
 
 \begin{lemma}
  \normalfont (RC3) \textit{is $(\mathcal F,4)$-boundary-reducible.}
 \end{lemma}
 
 \begin{proof}
 Let $v$ be a $d$-vertex adjacent to $d-2$ vertices of degree 3 where $v$ is either on a $3$-cycle or is adjacent to a vertex on a $3$-cycle.
 Denote the $3$-cycle as $w_1w_2w_3$, let $A$ be the set of $3$-vertices in $N(v)$.
 Let $H$ be the graph on $A\cup\{v\}\cup\{w_1,w_2,w_3\}$ and set the boundary  $B=V(H)\setminus (A\cup \{v\})$.
 
 Note that $H-B$ is isomorphic to either $K_{1,d-2}$ or $K_{1,d-2}$ with one additional edge. 
Observe that every vertex $x$ of $H-B$ has two available colors if $x\in A\cup\{v\}$ and has three available colors if $x\in\{w_1, w_2, w_3\}\setminus\{v\}$. 
 
 (FIX): If $H-B$ is a star, then fixing a color of an arbitrary vertex of $H-B$ can be extended to all of $H-B$ since all vertices have at least two available colors.
 If $H-B$ contains a $C_3$, then fixing a color of an arbitrary vertex of the $C_3$ can be extended to all of $H$ since the vertices in $\{w_1, w_2, w_3\}\setminus\{v\}$ have at least three available colors.

 (FORB): Let $I \subset V(H) \setminus B$ where $|I|\leq 2$. 
 If $|I|=2$, then $I$ is not $\mathcal{F}$-forbidding since connecting a new vertex to any pair of vertices in $H$ always creates either two $C_3$ at distance at most $1$ or a $4$-cycle.
It remains to consider the cases where $|I|=1$ but that is implied by (FIX). 
 \end{proof}

For each vertex $v$ and each face $f$, let $ch(v)=\deg(v)-2$ and $ch(f)=-2$.
By Euler's formula the sum of initial charge is negative: $\sum_{v\in V(G)}(\deg(v)-2)+\sum_{f\in F(G)}(-2)=2|E(G)|-2|V(G)|-2|F(G)|=-4$.
Consider a vertex $v$ that is adjacent to a $3$-vertex $w$. If $w$ is on a $3$-face $f$ that is not incident with $v$, then $f$ is a {\it pendent $3$-face} of $v$. 
The discharging rules are the following: 
\begin{itemize}
    \item [(D1)] Every 3-vertex sends charge 1/3 to each incident face.
    \item [(D2)] Every $4^+$-vertex $v$ sends 2/3 to each incident $3$-face and sends 1/3 to each pendent $3$-face. 
    The remaining charge is uniformly distributed among its incident $4^+$-faces. 
\end{itemize}

We now check that each vertex and each face has non-negative final charge. 

\begin{lemma}
Each vertex has non-negative final charge. 
\end{lemma}
\begin{proof}
By (RC1), there is no $2^-$-vertex. 
If $v$ is a $3$-vertex, then $ch^*(v)=1-\frac{1}{3}\cdot3=0$ by (D1). 

Suppose $v$ is a $4^+$-vertex. 
If $v$ is on a $3$-face, then $v$ has no pendent $3$-faces since there are no $3$-cycles with distance at most $1$. 

Since $ch(v)= \deg(v)-2\geq 2$, $v$ has non-negative final charge. 
If $v$ is not on a $3$-face, then it has at most $\deg(v)-3$ pendent $3$-faces by (RC3). 
Since $\deg(v)-2\geq \frac{\deg(v)-3}{3}$, $v$ has non-negative final charge. 
\end{proof}

\begin{lemma}
Each face has non-negative final charge.
\end{lemma}
\begin{proof}
Note first that there are no $4$-faces since there are no $4$-cycles. 

Suppose $f$ is a $3$-face. 
Each $4^+$-vertex on $f$ sends $\frac{2}{3}$ to $f$ by (D2). 
If $v$ is a $3$-vertex on $f$, then the neighbor $w$ of $v$ not on $f$ is a $4^+$-vertex by (RC3), so $f$ is a pendent $3$-face of $w$. 
Thus, $v$ and $w$ send $\frac{1}{3}$ and $\frac{1}{3}$ to $f$ by (D1) and (D2), respectively. 
Therefore, each vertex on $f$ guarantees at least $\frac{2}{3}$ to be sent to $f$, so $ch^*(f)\geq -2+3\cdot\frac{2}{3}=0$.

Suppose $f$ is a $5$-face, and let $v_1,v_2,v_3,v_4,v_5$ be the vertices on $f$ in clockwise ordering. 
By (RC2), $f$ cannot be incident with four $3$-vertices. 
If $f$ is incident with at most one $3$-vertex, then $f$ is incident with at least four $4^+$-vertices. 
Since each $4^+$-vertex on $f$ sends at least 
$\frac{5}{12}$ to $f$ by (D2), $ch^*(f) \geq -2 + 4\cdot\frac{5}{12}+\frac{1}{3} = 0$.
Therefore, we may assume $f$ is incident with either two or three $3$-vertices. 

Suppose $v_4, v_5$ are $3$-vertices on $f$, so $v_1$ and $v_3$ are $4^+$-vertices by (RC2). 
Since two $3$-vertices $v_4$ and $v_5$ are adjacent to each other, every vertex at distance at most $1$ from either $v_4$ or $v_5$ cannot be on a $3$-cycle by (RC3).
Moreover, for $i\in\{1, 3\}$, if $v_i$ is a $4$-vertex (resp.\ $5$-vertex), then $v_i$ has no (resp.\ at most one) pendent $3$-face by (RC3), so it sends $\frac{1}{2}$ (resp.\ $\frac{8}{15}$) to $f$ by (D2). 
If $v_i$ is a $6^+$-vertex, then it sends $\frac{1}{2}$ to $f$ by (D2). 
Thus, each of $v_1$ and $v_3$ sends at least $\frac{1}{2}$ to $f$, so $ch^*(f) \geq -2 + 2\cdot\frac{1}{2}+3\cdot\frac{1}{3} = 0$.

Suppose $v_2, v_5$ are $3$-vertices on $f$, so all other vertices on $f$ are $4^+$-vertices. 
Note that a $4^+$-vertex is guaranteed to send at least $\frac{5}{12}$ to $f$. 
If $v_1$ is on a $3$-face, then it must be a $5^+$-vertex by (RC3).
Thus, $ch(f)\geq -2+\frac{7}{12}+2\cdot\frac{5}{12}+2\cdot\frac{1}{3}=\frac{1}{12}>0$.
Now assume $v_1$ is not on a $3$-face. 
By (RC3), if $v_1$ is a $4$-vertex (resp.\ $5$-vertex), then $v_1$ has no (resp.\ at most one) pendent $3$-face, so it sends $\frac{1}{2}$ (resp.\ $\frac{8}{15}$) to $f$ by (D2). 
If $v_1$ is a $6^+$-vertex, then it sends $\frac{1}{2}$ to $f$ by (D2). 
Thus, $v_1$ sends at least $\frac{1}{2}$ to $f$, so $ch^*(f) \geq -2 + \frac{1}{2}+2\cdot\frac{5}{12}+2\cdot\frac{1}{3}=0$.

Finally, if $f$ is a $6^+$-face, then each vertex on $f$ sends at least $\frac{1}{3}$ to $f$ by (D1) and (D2), so $ch^*(f)\geq -2+6\cdot\frac{1}{3}=0$.
\end{proof}

For completeness we include a proof of the following Observation~\ref{obs:degeneracy} used in Table~\ref{tab:flex-choose}.

\begin{obs}
\label{obs:degeneracy}
Every planar graph without $C_4$ such that the distance between $C_3$'s is at least 2 is 3-degenerate.
\end{obs}
\begin{proof}
Suppose to the contrary that there exists a planar graph with minimum degree at least 4 but neither $C_4$ nor $C_3$ distance at most 1.
We use a simple discharging argument without any reducible configurations.
Let the initial charge of each vertex $v$ and each face $f$ be $\deg(v)-4$ and $|f|-4$, respectively. 
Therefore only 3-faces have negative initial charge.
The discharging rules are the following (They will be applied in the order they are presented):
\begin{itemize}
  \item each $5^+$-face sends 1/5 to each incident vertex.
  \item each vertex sends all its charge to its incident 3-face if it exists.
\end{itemize}
The final charge of each vertex remains non-negative. 
Each $5^+$-face $f$ has non-negative final charge since $(|f|-4) / |f|) \ge 1/5$.
Let $f$ be a $3$-face and let $v$ be a vertex on $f$. 
Each face incident with $v$ except $f$ is a $5^+$-face, since there are neither $C_4$ nor $C_3$ distance at most $1$. 
Since the minimum degree is at least $4$, $v$ receives charge at least $\frac{3}{5}$, all of which is sent to $f$. 
Thus, the final charge of $f$ is at least $-1+9/5>0$.
\end{proof}

\section{Proof of Theorem~\ref{C4C5}}\label{sec:C4C5}

In this section, we prove Theorem~\ref{C4C5}. 
Let ${\mathcal F}=\{C_4,C_5, C_6\}$ and let $G$ be a counterexample with the minimum number of vertices. 
Fix a plane embedding of $G$ and note that by minimality, $G$ must be connected.  
Let $L$ be a list assignment on $V(G)$ where each vertex receives at least four colors. 
The following configurations cannot appear in $G$:
\begin{itemize}
  \item [(RC1)] A $2^-$-vertex.
  \item [(RC2)] A $3$-cycle incident with two $3$-vertices.
  \item [(RC3)] A $4$-vertex $v$ on two $(3, 4, 4)$-faces.
  \item [(RC4)] Two adjacent $4$-vertices on $(3, 4, 4)$-faces that are mutually disjoint.
\end{itemize}
\begin{figure}[H]
  \centering
  \includegraphics[height=3cm]{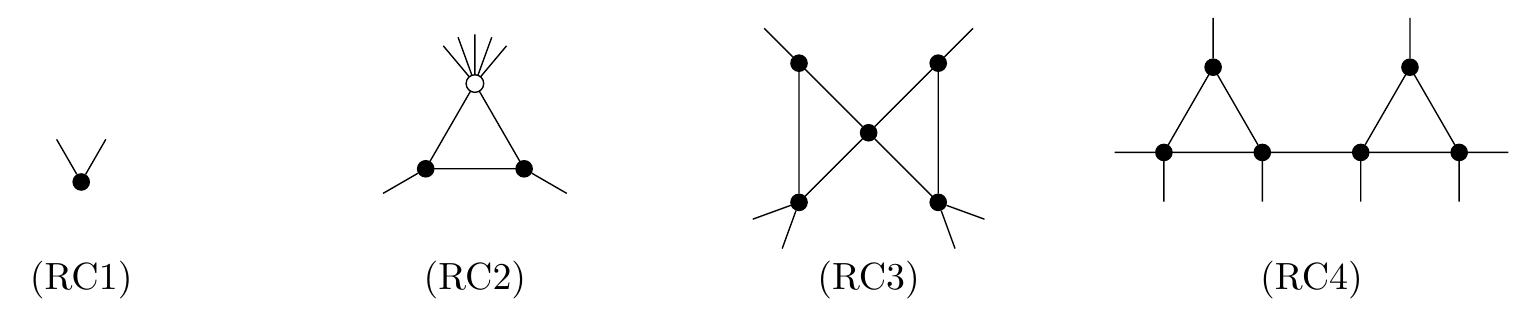}
  \caption{Reducible configurations}
\end{figure}

First, observe that (RC1) is reducible by Lemma~\ref{lem:singleVertex} and (RC2) is reducible by Lemma~\ref{lem:twoVerticesOnTriangle}.
 
\begin{lemma}
  \normalfont (RC3) \textit{is $(\mathcal \{C_4\},4)$-boundary-reducible with empty boundary.}
 \end{lemma}
 
 \begin{proof}
    Let $v$ be a $4$-vertex on two $3$-cycles $vu_1u_2$ and $vw_1w_2$, where $\deg(u_1)=\deg(w_1)=3$ and $\deg(u_2)=\deg(w_2)=4$.  
    Let $H$ be the graph on $u_1,u_2,v,w_1,w_2$ and set the boundary $B=\emptyset$. 
 
 (FIX): If vertex $v$ is fixed, then the coloring can be extended as $u_1$ and $w_1$ have lists of size 3 and $u_2$ and $w_2$ have lists of size 2. 
 Fixing any other vertex, say $u_1$, allows us to first color $u_2$, which then the leaves $v$ with at least two available colors. Therefore the coloring can be extended as  the list sizes of $w_1,w_2,v$ are $2,2,3$.  
 Fixing any other vertex in $H-v$ is handled in a similar fashion.  
 
 (FORB): Let $I \subset V(H) \setminus B$ where $|I|\leq 2$. 
 If $|I|=2$, then $I$ is not $\{C_4\}$-forbidding since connecting a new vertex to any pair of vertices in $H$ always creates a $C_4$. 
 It remains to consider the cases where $|I|=1$ but that is implied by (FIX).
 \end{proof}

 \begin{lemma}
   \normalfont (RC4) \textit{is $(\{C_4, C_5\},4)$-boundary-reducible with empty boundary.}
 \end{lemma}
 
 \begin{proof}
   Let $v_1,v_2$ be adjacent $4$-vertices on disjoint $3$-cycles $v_1u_1u_2$ and $v_2w_1w_2$, where $\deg(u_1)=\deg(w_1)=3$ and $\deg(u_2)=\deg(w_2)=4$.
 Let $H$ be the graph on $u_1,u_2,v_1,v_2,w_1,w_2$ and set the boundary $B=\emptyset$. 
 
 (FIX): If vertex $v_1$ is fixed, then we are left with two $3$-cycles.
The coloring can be extended since the two $3$-cycles have lists of sizes $1,2, 3$ and  $2,2, 3$.
If we fix any other vertex $x$, then the $v_i$ that lies on the same $3$-cycle as $x$ has at least one available color and therefore the coloring can be extended.

 (FORB): Let $I \subset V(H) \setminus B$ where $|I|= 2$. 
Only subsets of $\{v_1, v_2\}$ are $\{C_4, C_5\}$-forbidding.
In that case, we assign $v_1$ and $v_2$ distinct colors, leaving the two vertices on each $3$-cycle with one and two available colors.  The coloring can be extended by coloring the vertices with one available color first, then the remaining vertex on each $3$-cycle. 
 
 As was the case previously, the case where $|I|=1$ is implied by (FIX).
\end{proof} 
 
For each vertex $v$ and each face $f$, let $ch(v)=-2$ and $ch(f)=|f|-2$. 
By Euler's formula the sum of initial charge is negative: $\sum_{v\in V(G)}(-2)+\sum_{f\in F(G)}(|f|-2)=-2|V(G)|+2|E(G)|-2|F(G)|=-4$.
Note that there are no $2^-$-vertices by (RC1), so there are no $4$-faces and no $5$-faces. 
We remark that even though $C_6$ is forbidden, there might still be $6$-faces; these can appear only in the form of two embedded $3$-cycles. 

The discharging rules are as follows: 

\begin{itemize}
  \item [(D1)] Every $6^+$-face uniformly distributes its initial charge to every incident vertex.
  \item [(D2)] Let $f$ be a $3$-face. 
  \begin{itemize}
  \item [(D2A)] If $f$ is incident with a $3$-vertex, then $f$ sends $4/7$ to each incident $3$-vertex and uniformly distributes its remaining charge to each incident $4$-vertex. 
  \item [(D2B)] If $f$ is not incident with a $3$-vertex, then $f$ sends $3/7$ to each incident $4$-vertex on a $(3, 4, 4)$-face, sends $1/7$ to each incident $4$-vertex on a $(3, 4, 5^+)$-face, and sends $2/7$ to each incident $4$-vertex on two $(4^+, 4^+, 4^+)$-faces.
  \end{itemize}
\end{itemize}

Note that a $6$-face and a $7^+$-face sends $2/3$ and at least $5/7$, respectively, to each incident vertex by (D1).
Also, since each $3$-face is incident with at most one $3$-vertex by (RC2), it sends at least $1/7$ to each incident $4$-vertex on two 3-faces by (D2). 

We now check that each vertex and each face has non-negative final charge.  

\begin{lemma}
Each vertex has non-negative final charge. 
\end{lemma}
\begin{proof}
If $v$ is a $5^+$-vertex, then it is incident with at least three $6^+$-faces, so $ch^*(v)\geq -2+3\cdot\frac{2}{3}=0$ by (D1).
Suppose $v$ is a $3$-vertex. 
If $v$ is on a $3$-face, then $v$ is also on two $7^+$-faces, so $ch^*(v)\geq-2+2\cdot\frac{5}{7}+\frac{4}{7}=0$ by (D1) and (D2A).
If $v$ is not on a $3$-face, then $v$ is on three $6^+$-faces, so $ch^*(v)\geq-2+3\cdot\frac{2}{3}=0$ by (D1).

Suppose $v$ is a $4$-vertex.
If $v$ is on at most one 3-face, then it is on at least three $6^+$-faces, so $ch^*(v)\geq-2+3\cdot\frac{2}{3}=0$ by (D1). 
Now assume $v$ is on two $3$-faces $f_1$ and $f_2$. 
Note that $v$ is also on two $7^+$-faces, which each sends $5/7$ to $v$ by (D1). 
If both $f_1$ and $f_2$ are $(4^+, 4^+, 4^+)$-faces, then each of $f_1$ and $f_2$ sends $2/7$ to $v$, so $ch^*(v)\geq -2+2\cdot\frac{5}{7}+2\cdot\frac{2}{7}=0$. 
If $f_1$ is a $(3, 4, 5^+)$-face, which sends $3/7$ to $v$ by (D2A), then $f_2$ sends at least $1/7$ to $v$ by (D2), so $ch^*(v)\geq -2+2\cdot\frac{5}{7}+\frac{3}{7}+\frac{1}{7}=0$. 
By (RC3), the only remaining case to consider is when $f_1$ is a $(3, 4, 4)$-face and $f_2$ is a $(4^+, 4^+, 4^+)$-face. 
In this case, $f_1$ and $f_2$ send $3/14$ and $3/7$, respectively, to $v$, so $ch^*(v)\geq -2+2\cdot\frac{5}{7}+\frac{3}{14}+\frac{3}{7}>0$. 
\end{proof}

\begin{lemma}
Each face has non-negative final charge.
\end{lemma}
\begin{proof}
The final charge of every $6^+$-face is $0$ since its initial charge is positive and it distributes its initial charge uniformly to all incident vertices by (D2A). 
Recall that there are no $4$-faces and no $5$-faces. 

Let $f$ be a $3$-face, whose initial charge is $1$. 
Recall that $f$ is incident with at most one $3$-vertex by (RC2). 
If $f$ is incident with a $3$-vertex, then it has non-negative final charge by (D2A). 
If $f$ is not incident with a $3$-vertex, then by (RC4), it sends $3/7$ to at most one vertex by (D2B).
Thus, $ch^*(v)\geq1-\frac{3}{7}-2\cdot\frac{2}{7}=0$.
\end{proof}

\section{Proof of Theorem~\ref{weak}}\label{sec:weak}

In this section, we prove Theorem~\ref{weak}. 
Let ${\mathcal F}=\{\house, K_{2,3}\}$ and let $G$ be a counterexample with the minimum number of vertices. 
Fix a plane embedding of $G$ and note that by minimality, $G$ must be connected.  
Let $L$ be a list assignment on $V(G)$ where each vertex receives at least five colors. 
The following configurations cannot appear in $G$:

  \begin{itemize}
    \item [(RC1)] A $3^-$-vertex.
    \item [(RC2)] A $4$-cycle with two adjacent $4$-vertices. 
    \item [(RC3)] A $4$-cycle with consecutive vertices of degrees $4,5,4$. 
  \item [(RC4)] A $4$-cycle incident with three $5$-vertices and one $4$-vertex.
    \item [(RC5)] A path on four $4$-vertices. 
\end{itemize}

 \begin{figure}[H]
  \centering
  \includegraphics[width=\textwidth]{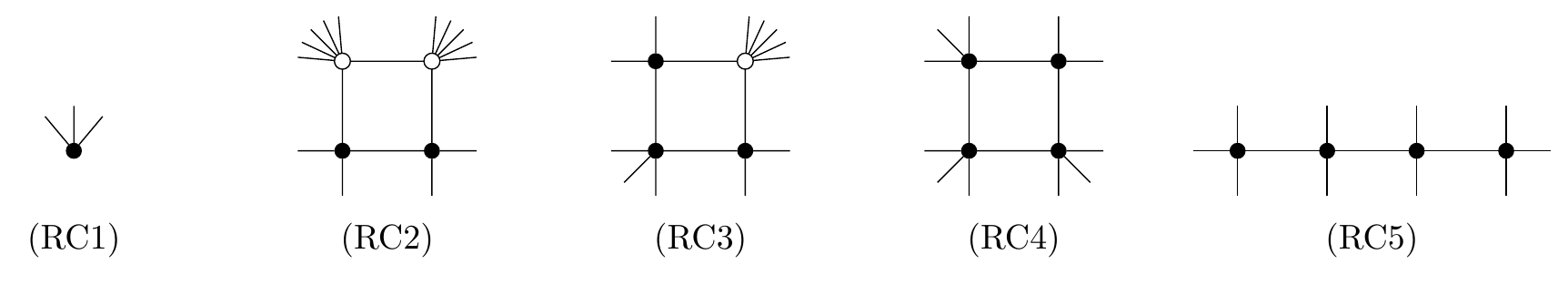}
  \caption{Reducible configurations}
  \label{weakred}
\end{figure}

First, observe that (RC1) is reducible by Lemma~\ref{lem:singleVertex}.
We establish the reducibility of the remaining configurations in the following lemmas.  
Note that (RC4) in Figure \ref{weakred} is only weakly $(\mathcal F,5)$-boundary-reducible but not $(\mathcal F,5)$-boundary-reducible, as the 4-vertex on the face does not satisfy (FIX).

    \begin{lemma}
      \normalfont A 4-cycle $v_1v_2v_3v_4$ with two adjacent $4$-vertices $v_1$ and $v_2$, (RC2), is $(\{\house\},5)$-boundary-reducible with boundary $\{v_3,v_4\}$.
 \end{lemma}
   \begin{proof}
   Let $v_1v_2v_3v_4$ be a $4$-cycle with $\deg(v_1)=\deg(v_2)=4$.
   Let $H$ be the graph on $v_1, v_2, v_3, v_4$ and set the boundary $B=\{v_3, v_4\}$.
 
 (FIX): 
 The coloring can be extended as an edge with lists of size $1$ and $2$ is clearly $L$-colorable. 
 
  (FORB): 
  Let $I\subset V(H)\setminus B$ where $|I|\leq 3$. 
  If $|I|\geq 2$, then $I$ is not $\mathcal \{C^+_5\}$-forbidding since connecting a new vertex to any pair of vertices in $H$ always creates a $\house$. 
  It remains to consider the cases where $|I|=1$ but that is implied by (FIX).   
 \end{proof}

    \begin{lemma}
      \normalfont A $4$-cycle $v_1v_2v_3v_4$ with consecutive vertices of degrees $4,5,4$ where $v_4$ has no degree restriction, (RC3), is $(\mathcal F,5)$-boundary-reducible with boundary $\{v_4\}$.
 \end{lemma}
 \begin{proof}
Let $v_1v_2v_3v_4$ be a $4$-cycle with $\deg(v_1)=\deg(v_3)=4$ and $\deg(v_2)=5$.
Let $H$ be the graph on $v_1, v_2, v_3, v_4$ and set the boundary $B=\{v_4\}$. 

 (FIX): If we fix the color of one vertex, then the coloring can be extended as the other two vertices have at least two available colors. 
 
 (FORB): Let $I \subset V(H) \setminus B$ where $|I|\leq 3$. 
 If $|I|\ge2$, then $I$ is not $\mathcal{F}$-forbidding, since connecting a new vertex to any pair of vertices in $H$ creates either a $\house$ or a $K_{2, 3}$.
 It remains to consider the cases where $|I|=1$ but that is implied by (FIX).
 \end{proof}

    \begin{lemma}
      \normalfont (RC4) \textit{is weakly $(\mathcal F,5)$-boundary-reducible with empty boundary.}
 \end{lemma}
 \begin{proof}
     Let $v_1v_2v_3v_4$ be a $4$-cycle with $\deg(v_1)=4$ and $\deg(v_2)=\deg(v_3)=\deg(v_4)=5$.
     Let $H$ be the graph on $v_1, v_2, v_3, v_4$ and set the boundary $B=\emptyset$.
 
 (FIX):
  Since our aim is to show weak reducibility it is sufficient to show (FIX) only for one vertex. 
  If we fix $v_i$ where $i\neq 1$, then we obtain a $4$-cycle where the list sizes are $1, 2, 2, 3$.

 (FORB): Let $I \subset V(H) \setminus B$ where $|I|\leq 3$. 
 If $|I|\ge2$, then $I$ is not $\mathcal{F}$-forbidding, since connecting a new vertex to any pair of vertices in $H$ creates either a $\house$ or a $K_{2, 3}$.
 It remains to consider the cases where $|I|=1$
 The only case that was not implied by (FIX) is  $I=\{v_1\}$, then we obtain a $4$-cycle where every vertex has two available colors. 
 \end{proof}
    
 \begin{lemma}
      \normalfont (RC5) \textit{is $(\emptyset,5)$-boundary-reducible with empty boundary.}
 \end{lemma}
  
 \begin{proof}
   Let $v_1,v_2,v_3,v_4$ be consecutive $4$-vertices on a path.
   Let $H$ be the graph on $v_1,v_2,v_3,v_4$ and set the boundary $B=\emptyset$. 
 
 (FIX):  If we fix the color of one vertex $v$, then the coloring can be extended to the remaining vertices on the path as the remaining vertices have at least two available colors prior to fixing the color of $v$.  
 
 (FORB): Let $I \subset V(H) \setminus B$ where $|I|\leq 3$. 
 Up to symmetry it suffices to check the cases $I=\{v_1,v_2,v_3\}$ and $I=\{v_1,v_2,v_4\}$.  
 In the first and second case, we are left with a path with lists of size $1,2,2,2$ and $1,2,3,1$, respectively, in this order. 
 In both cases, the coloring can be extended. 
 \end{proof}

 For each vertex $v$ and each face $f$, let $ch(v)=\deg(v)-2$ and $ch(f)=-2$. 
 By Euler's formula the sum of initial charge is negative: $\sum_{v\in V(G)}(\deg(v)-2)+\sum_{f\in F(G)}(-2)=2|E(G)|-2|V(G)|-2|F(G)|=-4$.
The discharging rules are as follows:

\begin{itemize}
  \item [(D1)] Every $6^+$-vertex sends $2/3$ to each incident face.
  \item [(D2)] Every $5$-vertex sends $2/3$ to each incident $3$-face and sends $1/2$ to each incident $4^+$-face.
  \item [(D3)] Every $4$-vertex sends $2/3$ to each incident $3$-face and sends $1/3$ to each incident $4^+$-face.
\end{itemize}

We now check that each vertex and each face has non-negative final charge.

\begin{lemma}
Each vertex has non-negative final charge. 
\end{lemma}
\begin{proof}
  There are no $3^-$-vertices by (RC1).
  If $v$ is a $6^+$-vertex then $ch^*(v)= \deg(v)-2-\deg(v)\cdot\frac{2}{3}\geq0$.
  If $v$ is a $5$-vertex, then it is adjacent to at most three $3$-faces, since there is no $\house$, so $ch^*(v)\geq 3-3\cdot\frac{2}{3}-2\cdot\frac{1}{2}=0$.
  If $v$ is a $4$-vertex, then it is adjacent to at most two $3$-faces, since there is no $\house$, so $ch^*(v)\geq 2-2\cdot\frac{2}{3}-2\cdot\frac{1}{3}=0$.
\end{proof}

\begin{lemma}\label{first_redistribution}
  Each face has non-negative final charge. 
\end{lemma}
\begin{proof}
  If $f$ is a $3$-face, then in all situations each vertex on $f$ gives $2/3$ to $f$, so $ch^*(f)=-2+3\cdot\frac{2}{3}=0$. 
   If $f$ is a $4$-face, then $f$ is incident with at most two $4$-vertices by (RC2). 
   We distinguish the following cases based on the number of $4$-vertices incident with $f$. 
  \begin{itemize}
    \item If there are exactly two $4$-vertices, then by (RC2) the $4$-vertices cannot be consecutive on $f$. 
    By (RC3) the other two vertices are $6^+$-vertices. Therefore, $ch^*(f)= -2+2\cdot\frac{2}{3}+2\cdot\frac{1}{3}=0$.
    \item If there is exactly one $4$-vertex, then $f$ is incident with at least one other $6^+$-vertex by (RC4). 
    Therefore, $ch^*(f)\geq -2+\frac{2}{3}+2\cdot\frac{1}{2}+\frac{1}{3}=0$.
    \item If there is no $4$-vertex, then all other vertices have degree at least $5$. Therefore, $ch^*(f)\geq -2+4\cdot\frac{1}{2}=0$.
  \end{itemize}
  If $f$ is a $5^+$-face, then $f$ is incident with five $4^+$-vertices, two of which are $5^+$-vertices, by (RC5), so $ch^*(f)\geq -2+2\cdot\frac{1}{2}+3\cdot\frac{1}{3}=0$.
\end{proof}

\section{Conclusions}
One can see Lemma~\ref{lem:evenImproved} and, in particular, definition of (weak) $(\mathcal{F},k,b)$-resolution as a generalization of degeneracy order.
There, the role of single vertices is replaced by (weak) $(\mathcal{F},k)$-boundary-reducibile configurations.
The resolution can be easily constructed in polynomial time under a mild assumption that the size of the boundary of each reducible configuration is bounded (which is the case of the theorems it is applied to in this paper).
Using the resolution, the list coloring satisfying any request for a single vertex only can be obtained straightforwardly.
However, it is not clear how to reconstruct $\varepsilon$-satisifable coloring for the given request (using the resolution or not) 
even though the existence of such coloring is guaranteed by Lemma~\ref{lem:evenImproved}.
As the flexibility concept has a substantial algorithmic motivation, it would be very interesting to explore its algorithmic potential.

Besides the open questions given in the introduction, we propose some open areas of inquiry.  In addition to further exploring the notion of weak flexibility, we propose two possible directions that align with the general effort to distinguish between flexibility and choosability in the class of planar graphs.  

First, Cohen-Addad, Hebdige, Kr\'aľ, Li, and Salgado~\cite{steinberg-false} constructed a planar graph with neither $C_4$ nor $C_5$ that is not even $3$-colorable, refuting Steinberg's Conjecture (see~\cite{steinberg}).
In this vein, we feel it would be interesting to determine whether it is possible to strengthen Theorem~\ref{C4C5} to graphs without $C_4$ and $C_5$ with lists of size 4.

As pointed out in~\cite{Dvorakmain} it would be nice to narrow the gap between $d$-degenerate graphs and (weighted) $\varepsilon$-flexible graphs with lists of size $k$. Theorem~\ref{thm:degeneracy} shows that $d$-degenerate graphs are $\varepsilon$-flexible when the size of the lists is at least $d+2$, but the bound conjectured in~\cite{Dvorakmain} is $d+1$. We propose a study of outer-planar graphs, which are $2$-degenerate, but it is not known whether lists of size $3$ suffice to achieve $\varepsilon$-flexibility.

\section*{Acknowledgement}

This work was completed in part at the 2019 Graduate Research Workshop in Combinatorics,
which was supported in part by National Science Foundation grants \#1923238 and \#1604458, National Security Agency grant \#H98230-18-1-0017, a generous
award from the Combinatorics Foundation, and Simons Foundation Collaboration Grants \#426971
(to M. Ferrara) and \#315347 (to J. Martin).

Ilkyoo Choi is supported by the Basic Science Research Program through the National Research Foundation of Korea (NRF) funded by the Ministry of Education (NRF-2018R1D1A1B07043049), and also by the Hankuk University of Foreign Studies Research Fund.  Paul Horn is supported by Simons Foundation Collaboration Grant \#525039.  Part of this work was performed while he was enjoying the hospitality of the Yau Mathematical Sciences Center at Tsinghua University in Beijing, China.  Tom\'{a}\v{s} Masa\v{r}\'{i}k received funding from the European Research Council (ERC) under the European Union's Horizon 2020 research and innovation programme Grant Agreement 714704, and from Charles University student grant SVV-2017-260452.  Dr.\ Masa\v{r}\'{i}k completed a part of this work while he was a postdoc at Simon Fraser University in Canada.

\end{document}